%% file: neurips_2026.tex
\documentclass{article}

\PassOptionsToPackage{numbers,sort&compress}{natbib}

\usepackage[preprint]{neurips_2026}

\usepackage[utf8]{inputenc} 
\usepackage[T1]{fontenc}    
\usepackage{hyperref}       
\usepackage{url}            
\usepackage{booktabs}       
\usepackage{amsfonts}       
\usepackage{nicefrac}       
\usepackage{microtype}      
\usepackage{xcolor}         
\usepackage{comment}
\usepackage{tabularray}
\usepackage{adjustbox}
\usepackage{url}
\usepackage{multirow}
\usepackage{multicol}
\usepackage{subcaption}
\usepackage{amsmath}
\usepackage{amssymb}

\newtheorem{definition}{Definition}

\newtheorem{lemma}{Lemma}
\title{Privacy Without Losing Place: A Paradigm for Private Retrieval in Spatial RAGs}

\author{%
  Kennedy Edemacu\thanks{The City University of New York, College of Staten Island and Graduate Center} \\
  The City University of New York\\
  New York, NY 10314 \\
  \texttt{kennedy.edemacu@csi.cuny.edu} \\
   \And
   Mohammad Mahdi Shokri\thanks{The City University of New York, Graduate Center} \\
   The City University of New York \\
   New York, NY 10016 \\
   \texttt{mshokri@gradcenter.cuny.edu} \\
   \AND
   Vinay M.~Shashidhar \\
   Northern Michigan
University \\
    Marquette, Mi 49855 \\
   \texttt{vmadanbh@nmu.edu} \\
   \And
   Jong Wook Kim \\
   Sangmyung University \\
   Seoul, Korea \\
   \texttt{jkim@smu.ac.kr} \\
}

\begin{document}

\maketitle

\begin{abstract}
  
  This work introduces PAS -- \textbf{P}rivacy \textbf{A}nchor \textbf{S}ubstitution, a structured mechanism for enabling user location privacy in spatial retrieval-augmented generation (RAG) systems. Unlike conventional differential privacy methods that directly perturb user locations, PAS represents location with relative anchor encoding consisting of an anchor, direction bin, and distance bin, allowing seamless integration with modern RAG pipelines. We evaluate PAS on a synthetic urban dataset and show that it achieves impressive coarse privacy guarantees, with approximately 370-400m adversarial location error, while retaining more than half of the baseline retrieval performance. Despite the slight drop in retrieval performance, the downstream generation quality under PAS remains comparatively robust, indicating that large language models can compensate for imperfect spatial retrieval. Furthermore, we provide empirical analysis showing that PAS exhibits non-monotonic privacy-utility relationship with respect to privacy parameters. We attribute this to geometric bias induced by anchor discretization, making it different from continuous noise mechanisms such as geo-indistinguishability. Our results show that structured spatial representations offer a practical approach to privacy in location based reasoning in RAG systems.
\end{abstract}

\section{Introduction}
The emergence of retrieval-augmented generation (RAG) has transformed the performance of large language models (LLMs), especially for knowledge-intensive applications
\cite{lewis2020retrieval, kandpal2023large, gao2023retrieval}.
LLMs are known to exhibit critical limitations in applications that require access to vast, specialized, or up-to-date sources of factual information. This inherent struggle stems from their static parametric memory and their inability to incorporate information beyond their context window \cite{roberts2020much, petroni2019language}. RAG directly addresses this limitation during inference by integrating relevant external knowledge dynamically 
retrieved from external sources \cite{zhu2024rageval, izacard2020leveraging, borgeaud2022improving}. Thus, significantly improving performance on tasks such as domain-specific question answering, fact-checking, and few-shot learning \cite{ram2023context, shi2023replug}.

Recent advances have extended this idea to geographical contexts, leading to the coinage of the term \textit{Spatial RAG}.
In Spatial RAG, retrieval is conditioned on user location and spatial metadata \cite{yu2025spatial, martins2025vision, jung2025intelligent, ni2025tp, zajac2025unifying, liu2025geospatial, schneider2025distrag, ruan2025retrieval}. 
The importance of this approach is far reaching, spanning multiple location-based application scenarios such as location-aware tourism, urban navigation and mobility guides, and disaster and emergency response systems. For example, by leveraging its ability to interpret natural language questions, a Spatial RAG system can accurately answer location-based queries such as \textit{what attraction features are near me within a reasonable walking distance?} and  \textit{find restaurants north of my current location within 1 mile range.}

However, current Spatial RAG studies \cite{yu2025spatial, martins2025vision, jung2025intelligent, ni2025tp, zajac2025unifying, liu2025geospatial, schneider2025distrag, ruan2025retrieval} assume access to the exact user location as a query feature. This creates a privacy bottleneck, as the system can expose fine-grained location coordinates of users to the retrieval system. Consequently, this introduces significant risks, including leakage, logging, and possible secondary use outside the scope of the intended application. Crucially, revealing the exact location coordinates can unintentionally expose sensitive and personal information, such as a user's home address, workplace, daily routines, or religious/medical affiliations, thus creating potential for surveillance, profiling, and targeted attacks. Established location privacy methods such as geo-indistinguishability, spatial cloaking, k-anonymity, and grid-based tokenization \cite{kim2022privacy, kim2021survey} are not tailored for RAGs and may suffer from a critical dichotomy. They either over-protect location data, thereby harming utility, owing to introduction of too much noise or overly large cloaks, or under-protect location data, resulting in the leakage of fine-grained geometrical information that adversaries can triangulate. For example, to protect location privacy in the query, \textit{find restaurants north of my current location within 1 mile range} using geo-indistinguishability, the direct injection of noise to the location coordinates shifts locations globally and may result in relevant point of interests falling outside the query radius.

A key insight for Spatial RAGs is that the retrieval process does not depend on exact location coordinates. Rather, it often relies on establishing semantic proximity, such as \textit{near a hospital} or \textit{within 1 km south of the university}. This finding opens the door for the development of task-aware privacy methods. Instead of perturbing the raw location coordinates, the system can release a customized spatial representation tailored for the retrieval task, thereby reducing the privacy risk. This work attempts to achieve this goal. To our knowledge, this effort represents the first attempt to protect location privacy in Spatial RAGs.

We propose a novel approach, \textbf{P}ublic \textbf{A}nchor \textbf{S}ubstitution (PAS), which is architected upon the aforementioned principle. 
PAS avoids the release of exact location coordinates by implementing a three-stage transformation. \textbf{1)} the user location is probabilistically mapped to a public anchor point (such as a landmark or transit hub). \textbf{2)} the spatial relationship between the true location and this anchor is encoded into coarse bins of distance and direction. And \textbf{3)} the system generates a query for the retrieval mechanism that exclusively leverages the resulting triplets (anchor, distance bin, direction bin), thereby successfully protecting the location information in the input as shown in Figure \ref{fig:pas_arch}. Thus, in summary, PAS provides the following benefits:
\begin{itemize}
    \item Anchors the privacy noise in semantically meaningful locations, successfully preventing unrealistic and unhelpful perturbed points.
    \item Produces features that are sufficient for effective retrieval and are aligned with the corpus indexing strategy, thus minimizing utility loss.
    \item Enables formal privacy guarantees via the exponential mechanism over public anchors, ensuring a coarse provable secure system. 
\end{itemize}

\begin{figure}
    \centering
    \includegraphics[width=0.6\textwidth]{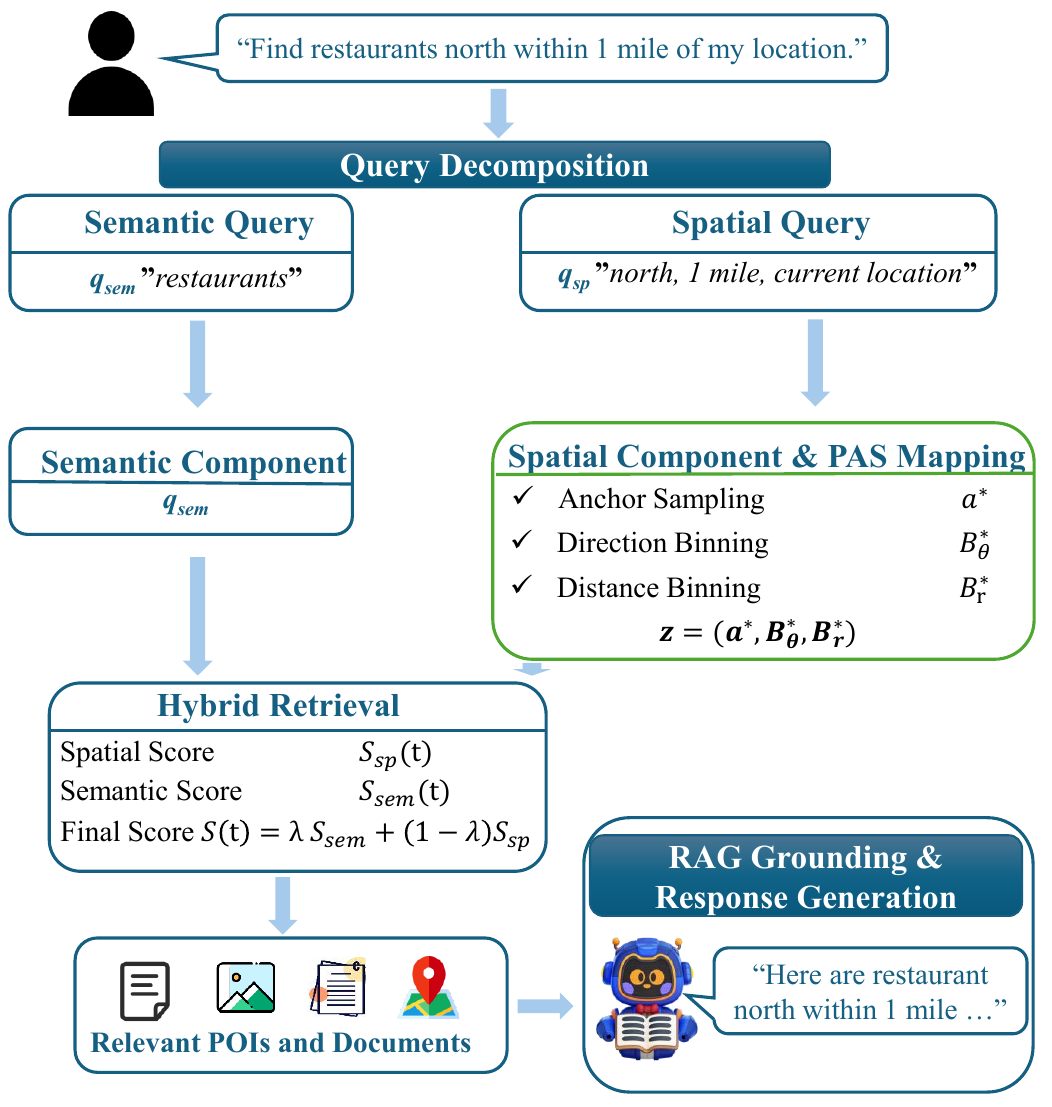}
    \caption{Illustration of PAS-based RAG framework. The main phases of the pipeline are: (1) query decomposition, where a user query is decomposed into semantic and spatial components. The spatial component undergoes a privacy mapping process, in which the true location is transformed into an obfuscated triplet (anchor, direction bin, distance bin). (2) the privatized spatial token is then combined with the semantic component to for a hybrid retrieval.  (3) the retrieved evidence is presented to RAG grounding model to the generate the final response.}
    \label{fig:pas_arch}
\end{figure}

\section{Methodology}\label{sec:methods}
We propose PAS, a location-obfuscation mechanism designed specifically for Spatial RAG tasks as shown in Figure \ref{fig:pas_arch}. Unlike traditional obfuscation techniques that release noisy location information, PAS maps a user's location in a semantically meaningful manner, referred to as \textit{anchor triplets}, which align with the corpus index while providing coarsely provable differential privacy guarantees.

\subsection{Problem Formalization}
\textbf{Spatial and Semantic Definitions}:
Given a user query, PAS first decomposes the query into spatial ($q_{sp}$) and semantic ($q_{sem}$) components. For example, for the query \textit{find restaurants north of my current location within 1 mile
range}, \textit{restaurant} is a semantic component, and \textit{current location}, \textit{north} and \textit{1 mile} form the spatial component. It is only the spatial part that gets privatized. PAS transforms the spatial part into (anchor, distance bin, direction bin) triplets. Formally, we let $\mathcal{U}$ represent a user's location space and define the privacy mapping components of PAS as follows.
\begin{itemize}
    \item Public Anchors ($A$): A finite set $A = \{a_1, \cdots, a_m\} \subset \mathcal{U}$ of publicly known landmarks such as campus buildings, transit hubs, and civic hubs.  
    \item Directional Partition ($B_{\theta}$): A finite partition of directional bearings, such as the 8 campus sectors. 
    \item Distance Partition ($B_{r}$): A finite partition of non-negative distances into discrete rings (e.g., $[0,500), [500, 2000), \cdots $ meters). 
\end{itemize}

\textbf{Corpus Indexing}:
We assume each document $t\in \mathcal{D}$ in a Spatial RAG corpus $\mathcal{D}$ is indexed by $\Phi(t) \in A \times B_{\theta} \times B_{r}$. In other words, each document in the corpus is tagged with the nearest or most relevant public anchor and coarse direction and distance bin tags relative to that anchor. This indexing strategy allows for pruning of relevant point of interests and ensures that retrieval is performed over the same discrete space as our obfuscated queries.

\subsection{PAS Privacy Mapping Mechanism}\label{sec:pas_mapping}
Given a true location $u\in \mathcal{U}$, a privacy parameter $\varepsilon > 0$, and a scaling factor $s$, for every location-based query, instead of exposing $u$, the $PAS_{\varepsilon, s}$ mechanism releases a sanitized query token $z = (a^*, B_{\theta}^{*}, B_{r}^{*}) \in A \times B_{\theta} \times B_{r}$ to the RAG for retrieval. The mechanism comprises two stages:

\textbf{(1) Probabilistic Anchor Selection}:
We take advantage of the Exponential Mechanism \cite{mcsherry2007mechanism} to probabilistically select a representative public anchor $a^* \in A$. The selection probability is proportional to the distance between the true location and the candidate anchor location, as shown in Eq. \ref{eq:selection}. 
\begin{equation}\label{eq:selection}
    \Pr(a^* \mid u) \propto \exp\left(-\varepsilon \frac{d(u, a^*)}{s}\right), \quad \forall a^* \in A
\end{equation}
This ensures that the true location is likely to be mapped to a nearby public anchor but maintains a strictly defined probability of being mapped to any public anchor in the set.

\textbf{(2) Deterministic Post-Processing}:
Once $a^*$ is identified and sampled, we generate the coarse semantics: distance and direction bins relative to the true location $u$.

\begin{enumerate}
    \item Directional Binning: Let $\theta(a^*, u)$ be the bearing from the true location $u$ to the sampled public anchor $a^*$. We set the directional bin $B_{\theta}^{*} =\text{\textit{dir\_bin}}(\theta(a^*, u)) \in B_{\theta}$.

    \item Distance Binning: Similarly, let $r(u, a^*) = d(u, a^*)$ be the distance between the true location $u$ and the sampled public anchor $a^*$. We set the distance bin as $B_{r}^{*} = \text{\textit{dist\_bin}}(r(u, a^*)) \in B_{r} $.
\end{enumerate}
The triplet $z = (a^*, B_{\theta}^{*}, B_{r}^{*})$ is what the RAG backend will gain access to, and the raw location $u$ never leaves the user's site. Most importantly, the token $z=(a^*, B_{\theta}^*, B_{r}^{*})$ defines an uncertainty region, i.e., the user lies in the set:
\begin{equation}
    U(z) = \{x: \theta(a^*,x)\in B_{\theta}^*, d(a^*,x)\in B_{r}^{*}\}
\end{equation}
For example, suppose $B_{\theta}^{*} = \text{north}$ and $B_{r}^{*}=\text{0-1mi}$, then $U(z)$ is:
\begin{equation*}
    U(z) = (\text{north wedge}) \cap (0-1 \text{mile ring}).
\end{equation*}
This allows the retriever to treat $u$ as unknown, but is constrained to $U(z)$.

\subsection{Point of Interest Retrieval}
To retrieve a point of interest (POI), $t$, we assume a distribution, $\pi(x|z)$ over possible user locations (latent user locations) inside the uncertainty region. Then:
\begin{equation}
    \mathcal{R}(z) =\text{top-n}_{t\in C(z)}S_{sp}(t|z)
\end{equation}
where:
\begin{equation*}S_{sp}(t|z) = \Pr_{x\sim \pi(.|z)}[d(x,t)\leq \text{R } \wedge \theta(x, t)\in B_q],
\end{equation*}
$C(z)$ is a candidate set, $R$ and $B_q$  are the distance and direction constraints specified in the query. This implies that we retrieve POIs with high probability of meeting the user-relative constraints.

In implementation, integration over all $x\in U(z)$ is challenging. Instead, we approximate with sampling using the Monte Carlo approximation \cite{metropolis1949monte}. We sample $K$ latent user points, $x_1,\cdots,x_K \in U(z)$, then we estimate:
\begin{equation}\label{eq:monte}
S_{sp}(t|z) = \frac{1}{K} \sum_{k=1}^{K} \mathbf{1}[d(x_k, t) \le R] \cdot \mathbf{1}[\theta(x_k, t)\in B_q]
\end{equation}
For textual relevance between the user query and the retrieved documents, we compute the semantic score for each retrieved item $t$. We define the semantic score, $S_{sem}(t)$ as a cosine similarity between $q_{sem}$ and $t$: 
\begin{equation}
    S_{sem}(t) = \text{cosine}(q_{sem},t)
\end{equation}
This formulation follows the standard dense retrieval in RAG pipelines.
We then perform a hybrid ranking for each POI, $t$ as:
\begin{equation}\label{eqn:hybrid_score}
    S(t)=\lambda S_{sem}(t) + (1-\lambda)S_{sp}(t|z)
\end{equation}
where $0\leq \lambda \leq 1$ is a weighting parameter. Any item $t$ whose final score $S(t)$ falls within the top-$k$, is selected for retrieval. 

To improve retrieval efficiency, we first employ the distance tags to generate a superset of plausible POIs before performing the above ranking. A good candidate set is defined as:
\begin{equation}
    C(z) = \{t: d(a^*,t)\leq R+r_{max}\}
\end{equation}
where $r_{max}$ is the maximum radius in $B_{r}$. This follows the principle of triangle inequality \cite{cormen2022introduction}. Since $d(u,a^*)\leq r_{max}$, any POI within the distance $R$ from the user lies within the distance $R+r_{max}$ from the anchor. Such a pruning can also be done with direction bins.

\subsection{Privacy Analysis}
We measure the privacy of PAS using the differential privacy framework of geo-indistinguishability (geo-DP). 

\begin{definition}[$\varepsilon$-geo-DP over Anchors]
A mechanism $M: \mathcal{U}\rightarrow A\times B_{\theta}\times B_{r}$ satisfies $\varepsilon$-geo-DP with scale $s$ if for all $u, u^{\prime} \in \mathcal{U}$ and all measurable $O\subseteq \text{supp}(M)$ \cite{andres2013geo}:
\begin{equation}
    \frac{\text{Pr}(M(u)\in O)}{\text{Pr}(M(u^{\prime})\in O)} \leq \text{exp}\big(\varepsilon \frac{d(u, u^{\prime})}{s}\big)
\end{equation}
\end{definition}

\begin{lemma}
    The mechanism $\text{PAS}_{\varepsilon,s}$ satisfies $\varepsilon\text{-geo-DP}$.
\end{lemma}
\textbf{Proof Sketch}: The sampling of $a^*$ follows the standard exponential mechanism where the score function $q(u,a^*) = -d(u,a^*)/s$ has the metric sensitivity $d(u,u^{\prime})/s$. By the standard exponential bound and the closure of differential privacy under post-processing, the subsequent deterministic computations of $B_{\theta}^{*}$ and $B_{r}^{*}$ preserve this bound.

\section{Experiment}\label{sec:experiments}
\subsection{Dataset}
To rigorously evaluate PAS, we construct a synthetic yet realistic Spatial RAG benchmark dataset based on New York City. The dataset is designed to simulate location-aware queries under privacy constraints. The dataset includes a diverse set of urban point-of-interests (POIs) such as restaurants, hotels, hospitals, parks, etc, reflecting diversity across the five boroughs of New York City. The dataset has three main components: anchors, chunks (POIs), and queries. \textbf{Anchors} contains a set of 30 spatial anchors distributed across the city. Each anchor is defined by a geographic coordinate and semantic identifiers such as neighborhood, name, etc. The anchors serve as a basis for spatial abstraction in PAS through probabilistic spatial representation. \textbf{Chunks} forms the corpus, which contains 1010 POIs. Each POI represents a retrievable document. Each POI includes semantic attributes (such as name, category, and description), a geographic coordinate, and multi-anchor tags. Multi-anchor tagging allows each chunk to be pre-indexed with spatial relations relative to at least two anchors. Each anchor tag specifies a direction bin and a distance bin relative to the anchor. This multi-anchor indexing enables the probabilistic spatial matching required to evaluate PAS, and it is vital for candidate pruning. \textbf{Queries} forms a set of 423 evaluation queries of varying complexity. Each query is equipped with a true location (for baseline evaluations only), spatial constraints (distance and direction bins), ground-truth POIs, and a semantic category. The dataset supports two retrieval modes: (i) the deterministic baseline retrieval, where the true user location is used, and (ii) PAS retrieval, where probabilistic tokens are generated for spatial representation. We used ChatGPT \cite{chatgpt2026} during the dataset construction, specifically to assist in generating components of synthetic queries and metadata. All generated content was reviewed and validated.


\subsection{Metrics}
We evaluate across three dimensions, with each dimension having its own metric(s) as follows.

\textbf{Retrieval Metrics}:
To assess the retrieval ability, we employ two metrics: \textit{Recall@k} and \textit{nDCG@k}. Recall@k measures the fraction of relevant POIs that appear in the top-k retrieved results. Meanwhile, nDCG@k \cite{wang2013theoretical} measures how well the retrieval identifies the relevant POIs and how high they are ranked amongst the retrieved items.

\textbf{Generation Metrics}:
We assess the generation quality using two metrics: \textit{F1 score} and \textit{Citation Correctness}. F1 score measures the overlap between the generated answer and the ground truth answer. We employ it to reflect if degradation in retrieval due to privacy integration propagates to the generation phase. Citation correctness (Citation) evaluates whether the generated answers are properly grounded in the retrieved context. 

\textbf{Security Metrics}:
We introduce a security metric, \textit{Expected Localization Error (ALE)}, which measures the expected distance between the true user location $u$ and an adversary's location estimate $\hat{u}$ given the PAS token $z$:
\begin{equation}
    \text{ALE}(z) = \mathbb{E}[d(u, \hat{u}) \mid z]
\end{equation}
where $\hat{u}$ is defined as the centroid of the uncertainty region $U(z)$, and it is computed using the Monte Carlo sampling approximation as:
\begin{equation}
    \hat{u} = \frac{1}{K} \sum_{k=1}^{K} x_k, \quad x_k \sim U(z)
\end{equation}
We use ALE to capture the expected error of an adversary attempting to infer the user's true location from the privatized spatial token.

\subsection{RAG Settings}
We evaluate our proposed PAS within the standard RAG pipeline, where a query is first decomposed into a semantic and a spatial component. The decomposed components are then used to retrieve relevant items from the grounded corpus before passing them to the generation model to generate the final answer. We employ the hybrid scoring approach in Eq. \ref{eqn:hybrid_score} during the retrieval. Unless stated otherwise, we use a $\lambda$ value of 0.8, giving priority to semantic relevance while incorporating spatial constraints.

For semantic retrieval, we perform dense vector retrieval with two widely used models: MiniLM (sentence-transformers/all-MiniLM-L6-v2) \cite{reimers-2019-sentence-bert} and bge-base (bge-base-en) \cite{xiao2023bge}. Spatial relevance is computed using the true location for the baseline (non-private). However, for PAS, the true location is replaced with the triplets of a sampled anchor, and discretized direction and distance bins relative to the sampled anchor. The triplets define an uncertainity region for the user location. Thus, spatial relevance is estimated over this region instead of a single point. For each query, we generate a set of latent user locations from the uncertainity region, and employ Monte Carlo sampling over the latent locations to estimate the spatial relevance for corpus items. The selected top-k items are based on the hybrid score comprising the semantic and spatial relevance scores. We use a top-5 in our experiments. 

To improve efficiency, we reduce the search space by first pruning the documents using anchor-based constraints. For a DP sampled anchor, we retain documents that lie within an expanded radius relative to the anchor. The retrieved documents (POIs) are presented to generation models to output the final answer. We evaluate with three generation models: Meta-Llama-3.3-70B \cite{llama3modelcard}, gemma-3-12b-it \cite{gemmateam2025gemma3technicalreport}, and gpt-oss-120b \cite{openai2025gptoss120bgptoss20bmodel}. For other settings, we computed distances using the Haversin formula \cite{inman1835navigation} and we discretize bearings into 8 direction bins. Further details are provide in Appendix \ref{app:experiment_setup}.

\section{Results}\label{sec:results}
\begin{table*}[t]
\centering
\small
\begin{tabular}{llcccccc}
\toprule
Retriever & Model & Setting & Recall@k & nDCG@k & F1 & Citation & ALE (m) \\
\midrule

\multirow{6}{*}{\textbf{bge-base}} 
& \multirow{2}{*}{Meta-Llama-3.3-70B} 
& Baseline & \textbf{0.8456} & \textbf{0.7778} & 0.2428 & \textbf{0.8017} & 0.00 \\
& 
& PAS ($\epsilon=1.0$) & 0.5260 & 0.4623 & 0.2577 & 0.5667 & 367.57 \\

& \multirow{2}{*}{gemma-3-12b-it} 
& Baseline & \textbf{0.8456} & \textbf{0.7778} & \textbf{0.4686} & 0.7380 & 0.00 \\
& 
& PAS ($\epsilon=1.0$) & 0.5190 & 0.4533 & 0.4168 & 0.5310 & \textbf{389.65} \\

& \multirow{2}{*}{gpt-oss-120b} 
& Baseline & \textbf{0.8456} & \textbf{0.7778} & 0.3977 & 0.7840 & 0.00 \\
& 
& PAS ($\epsilon=1.0$) & 0.4508 & 0.4270 & 0.3625 & 0.5567 & 370.61 \\

\midrule

\multirow{6}{*}{\textbf{MiniLM}} 
& \multirow{2}{*}{Meta-Llama-3.3-70B} 
& Baseline & \textbf{0.8719} & 0.8089 & 0.2726 & 0.7451 & 0.00 \\
& 
& PAS ($\epsilon=1.0$) & 0.5807 & 0.5163 & 0.2673 & 0.6559 & 372.98 \\

& \multirow{2}{*}{gemma-3-12b-it} 
& Baseline & \textbf{0.8719} & 0.8089 & \textbf{0.4779} & 0.6690 & 0.00 \\
& 
& PAS ($\epsilon=1.0$) & 0.5504 & 0.4896 & 0.4300 & 0.5112 & \textbf{397.88} \\

& \multirow{2}{*}{gpt-oss-120b} 
& Baseline & \textbf{0.8719} & \textbf{0.8091} & 0.3566 & \textbf{0.7589} & 0.00 \\
& 
& PAS ($\epsilon=1.0$) & 0.5648 & 0.5322 & 0.3657 & 0.6012 & 369.12 \\

\bottomrule
\end{tabular}
\caption{Comparison of PAS and baseline (non-private) performance across retrievers. PAS introduces strong privacy (ALE $\approx$ 370--400m) while reducing retrieval performance. Best values are \textbf{bold}.
}
\label{tab:pas_results}
\end{table*}

\begin{figure*}[htbp]
     \centering
     \begin{subfigure}[b]{0.48\textwidth}
         \centering
         \includegraphics[width=\textwidth]{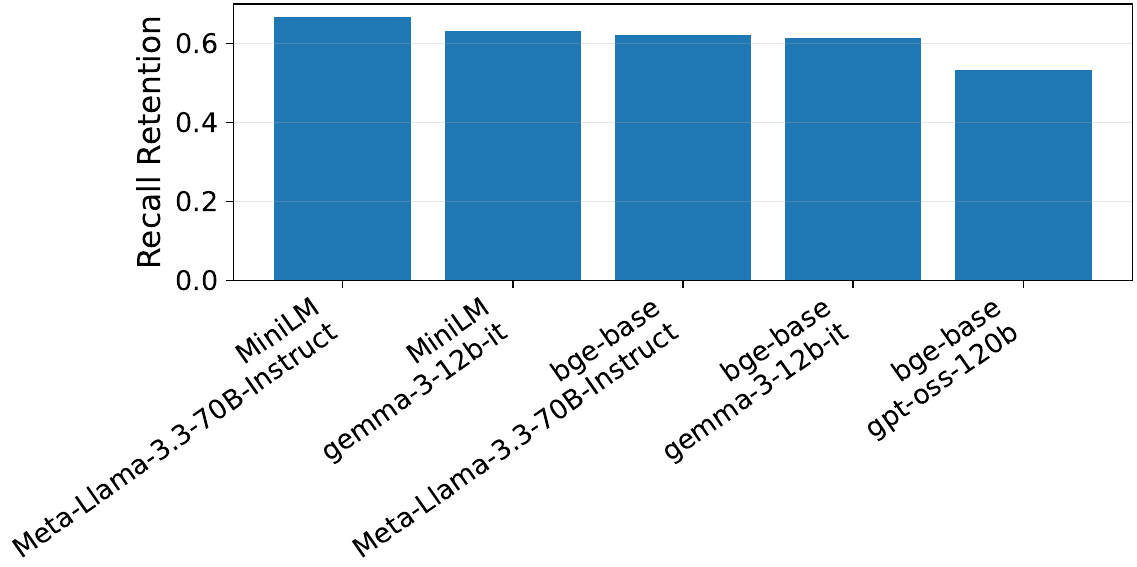}
         \caption{Recall@k Retention}
         \label{fig:ret_recall}
     \end{subfigure}
     \hfill
     \begin{subfigure}[b]{0.48\textwidth}
         \centering
         \includegraphics[width=\textwidth]{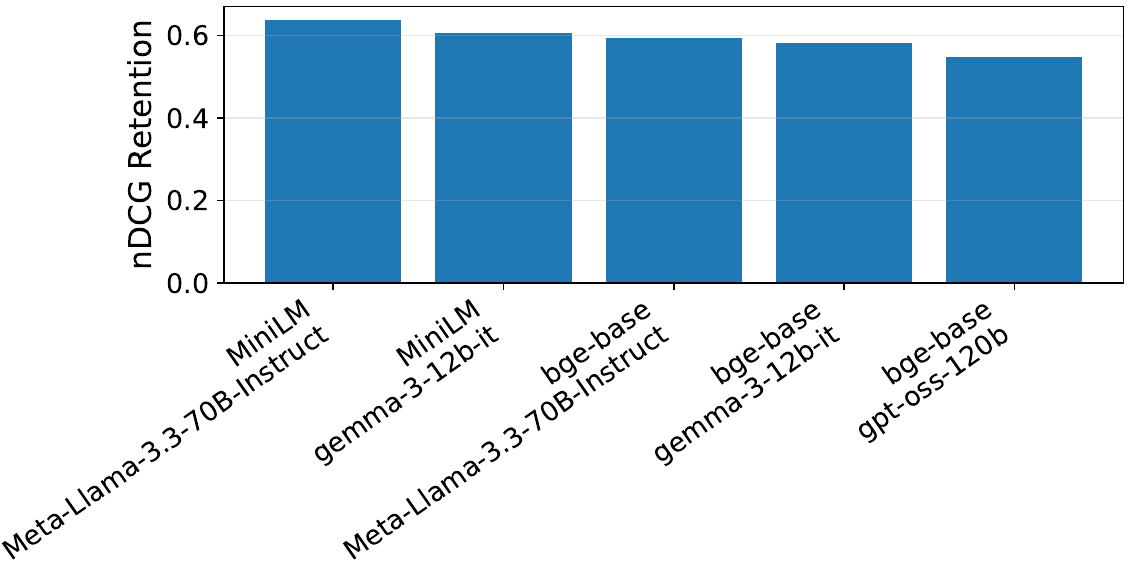}
         \caption{nDCG@k Retention}
         \label{fig:ret_ndcg}
     \end{subfigure}
     \hfill
     \begin{subfigure}[b]{0.48\textwidth}
         \centering
         \includegraphics[width=\textwidth]{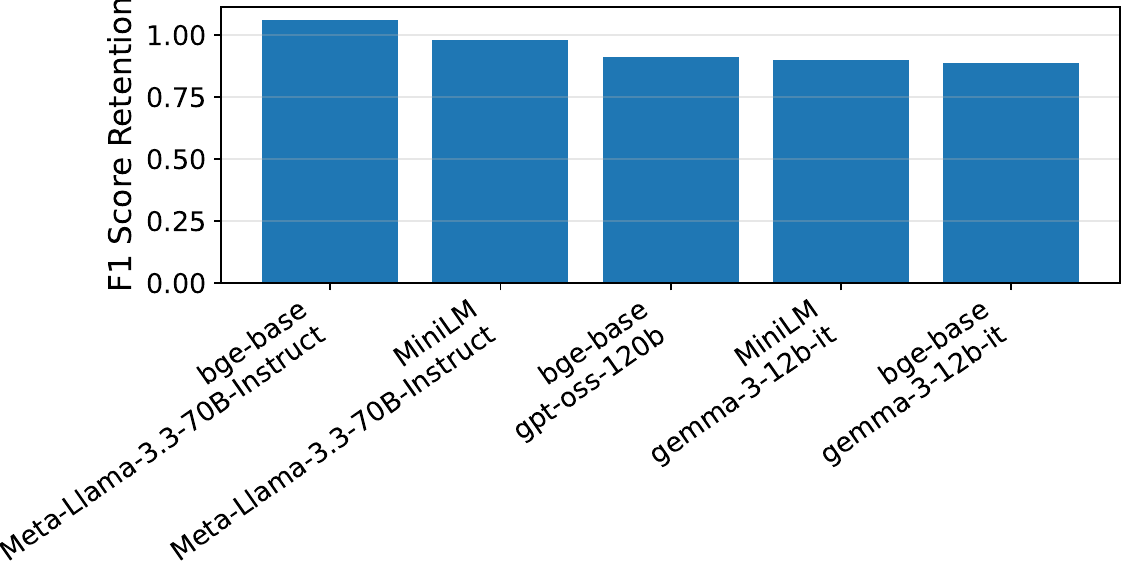}
         \caption{F1 Score Retention}
         \label{fig:ret_f1}
     \end{subfigure}
     \hfill
     \begin{subfigure}[b]{0.48\textwidth}
         \centering
         \includegraphics[width=\textwidth]{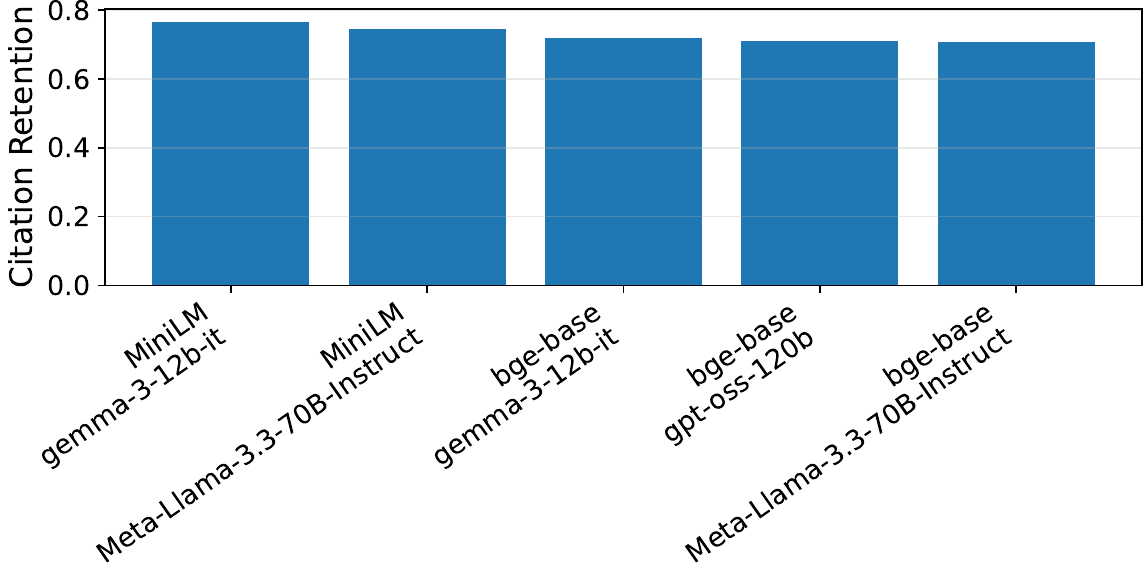}
         \caption{Citation Correctness Retention}
         \label{fig:ret_citation}
     \end{subfigure}
     \caption{Performance Retention under PAS}
     \label{fig:retention}
\end{figure*}

\begin{figure*}[htbp]
     \centering
     \begin{subfigure}[b]{0.48\textwidth}
         \centering
         \includegraphics[width=\textwidth]{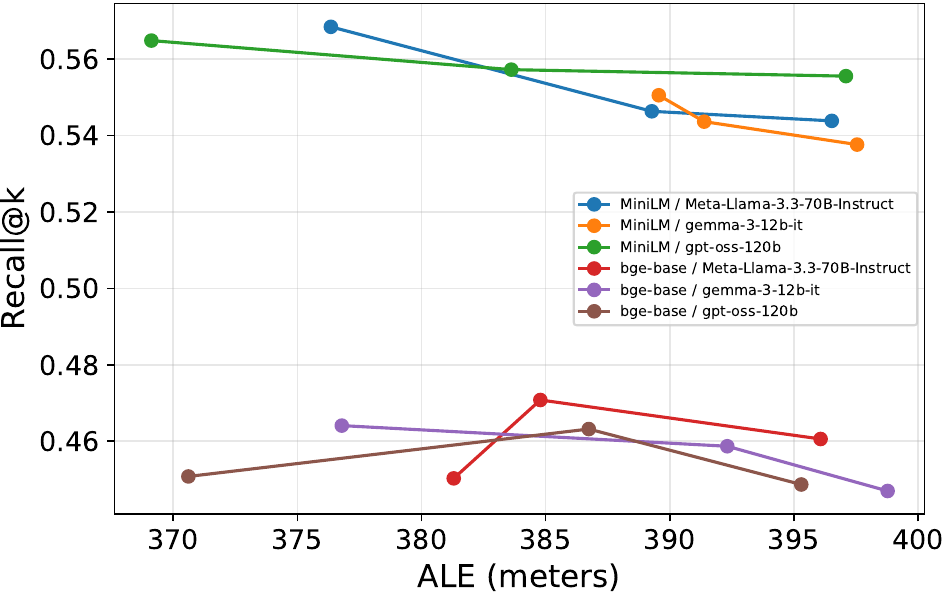}
         \caption{Privacy vs Recall}
         \label{fig:sub1}
     \end{subfigure}
     \hfill
     \begin{subfigure}[b]{0.48\textwidth}
         \centering
         \includegraphics[width=\textwidth]{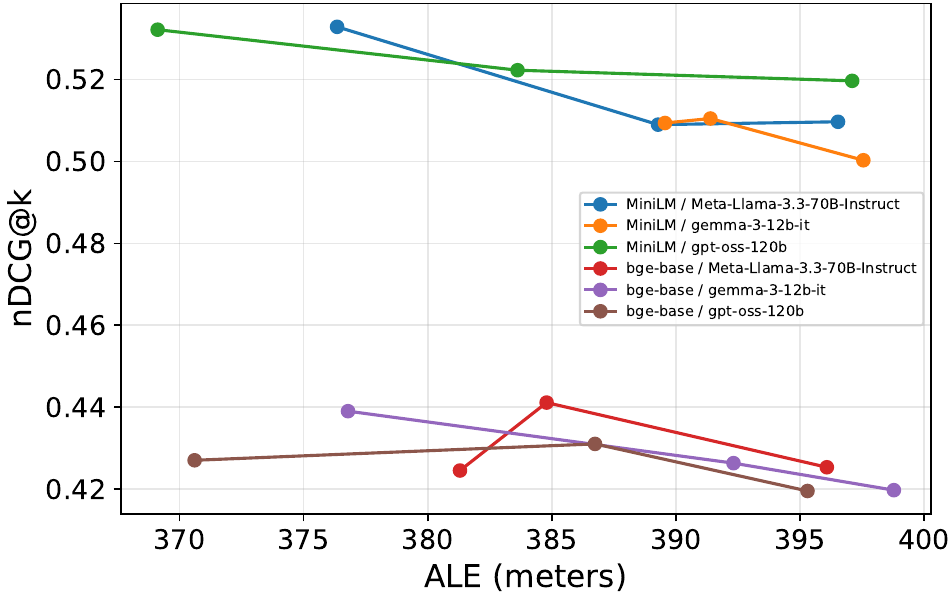}
         \caption{Privacy vs nDCG}
         \label{fig:sub2}
     \end{subfigure}
     \caption{Privacy-Utility Trade-off Curves for PAS}
     \label{fig:privacy-utility}
\end{figure*}

We evaluate the impact of PAS framework on retrieval and generation quality across two retrievers: bge-base and MiniLM, and three generation models: Meta-Llama-3.3-70B, gemma-3-12b-it and gpt-oss-120b. We summarize the results in comparison to the baseline (non-private) in Table \ref{tab:pas_results}, with additional graphical comparisons in Appendix \ref{app:pas_baseline}. In general, across all retriever-generation model configurations, PAS achieves strong privacy guarantees, with ALE consistently ranging between 367 and 398 meters, corresponding to neighborhood-level obfuscation. However, this coarse privacy guaranty comes at a measurable cost in retrieval performance. We further illustrate this through the privacy-utility trade-off curves presented in Figure \ref{fig:privacy-utility}. 

\subsection{Retrieval Performance Analysis}
We observe that increasing privacy (high ALE) results in decrease in retrieval performance. Specifically, from Table \ref{tab:pas_results} we observe that Recall@k decreases from $\sim$0.85-0.87 to $\sim$0.45-0.58, and nDCG@k drops from $\sim$0.78-0.81 to $\sim$0.43-0.53. Furthermore, we observe that the drop in retrieval performance is worse for bge-base compared to MiniLM, suggesting that MiniLM is more robust to spatial perturbations, possibly due to stronger semantic generalization. However, this degradation in retrieval performance is rather gradual than catastrophic, indicating that PAS preserves meaningful portion of spatial signals despite location obfuscation. To this effect, we report the retention of the retrieval performance in Figures \ref{fig:ret_recall} and \ref{fig:ret_ndcg}. Recall@k and nDCG@k retentions are measured as: $\frac{\text{PAS Recall@k}}{\text{Baseline Recall@k}}$ and $\frac{\text{PAS Recall@k}}{\text{Baseline Recall@k}}$, respectively. In Figure \ref{fig:ret_recall}, we observe that in most cases, PAS retains at least 60\% of the baseline Recall@k. Meanwhile in Figure \ref{fig:ret_ndcg}, we observe that PAS retains at least 55\% of the baseline nDCG@k. This shows that PAS maintains a substantial fraction of retrieval capabilities under privacy constraints.
Despite the drop in retrieval, the fairly high nDCG@k retention indicates that PAS preserves the relative ordering of the relevant documents even when some are missing. Primarily this is because PAS affects coverage rather than ranking fidelity. 

\subsection{Generation Performance Analysis}
Despite the decrease in retrieval performance, the generation quality remains fairly stable as shown in Table \ref{tab:pas_results}, and Figures \ref{fig:ret_f1} and \ref{fig:ret_citation}. F1 score exhibits moderate drop and even improves in some cases. For example, for the bge-Llama configuration, F1 slightly increases from 0.2428 to 0.2577. While for MiniLM-Llama configuration, F1 slightly drops from 0.2726 to 0.2673. This suggests that PAS noise may likely act as a regularizer, minimizing overfitting to precise spatial cues, thus encouraging reliance on semantic signals. Citation correctness exhibits degradation but remains largely at usable levels. Overall, the citation correctness drops from $\sim$0.74-0.80 in baseline to $\sim$0.51-0.66 in PAS. In general, PAS highly retains the generational performances of the baseline as shown in Figures \ref{fig:ret_f1} and \ref{fig:ret_citation}. This together with the retrieval vs generation results presented in Appendix \ref{app:retrieval_generation}, indicate that models can still produce quality answers even with imperfect retrieval. However, the overall final answer quality is dependent on the generation model used. Meta-Llama-3.3-70B achieves the most stable performance under PAS, while gemma-3-12b-it achieves the highest baseline F1 score but experiences the greatest degradation. This suggests that larger or better instruction-tuned models can compensate for the noisy retrieved inputs. 

\subsection{Privacy and Utility Sensivity to $\varepsilon$ and $\lambda$ Analysis}
\begin{figure*}[!t]
     \centering
     \begin{subfigure}[b]{0.48\textwidth}
         \centering
         \includegraphics[width=\textwidth]{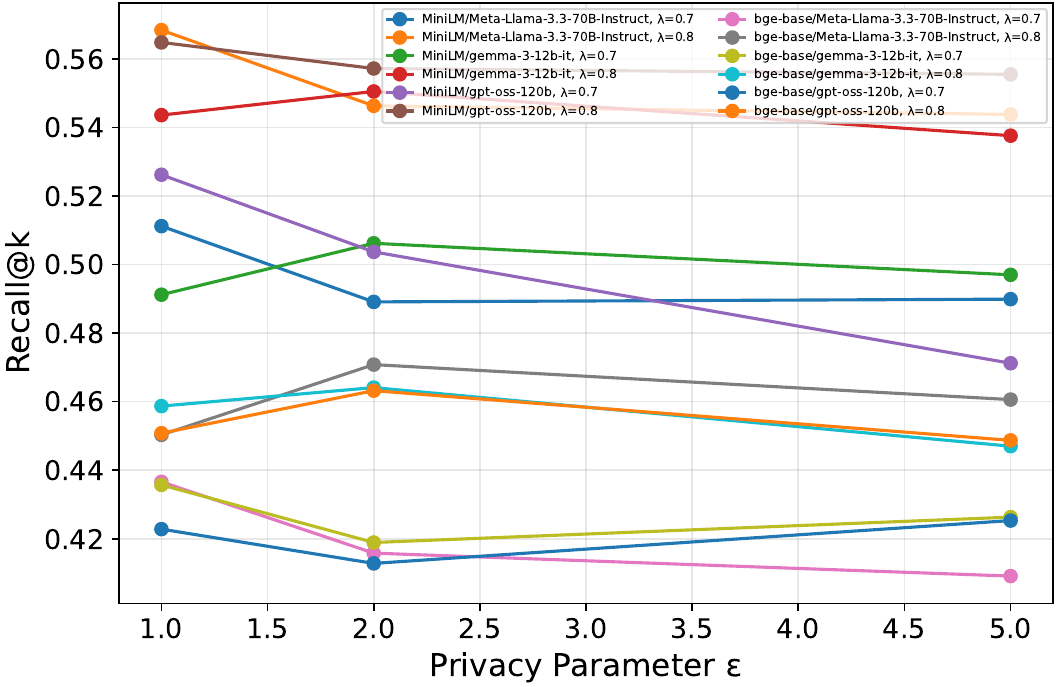}
         \caption{Recall Sensitivity to $\varepsilon$ and $\lambda$}
         \label{fig:recall-sensiti}
     \end{subfigure}
     \hfill
     \begin{subfigure}[b]{0.48\textwidth}
         \centering
         \includegraphics[width=\textwidth]{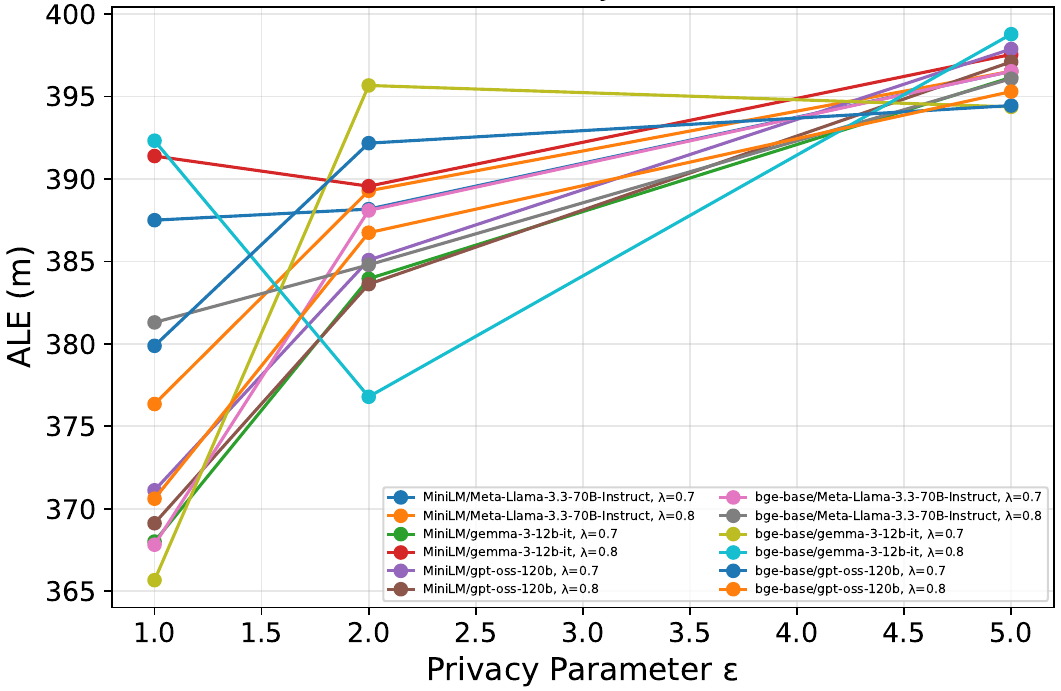}
         \caption{ALE Sensitivity to $\varepsilon$ and $\lambda$}
         \label{fig:ale-sensiti}
     \end{subfigure}
     \caption{Privacy and Utility Sensitivity to $\varepsilon$ and $\lambda$}
     \label{fig:privacy-utility-sensitivity}
\end{figure*}
We also report privacy and utility sensitivity to varying privacy parameters, as shown in Figure \ref{fig:privacy-utility-sensitivity}. Interestingly, we observe that an increase in $\varepsilon$ does not monotonically result in improved utility. Instead, we observe that Recall@k decreases or remains flat as $\varepsilon$ increases and ALE increases as $\varepsilon$ increases, which contradicts the typical differential privacy intuition. Unlike the traditional geo-indistinguishability mechanism, PAS demonstrates a non-monotonic relationship between $\varepsilon$ and utility. This is mainly due to anchor discretization and spatial encoding. Although larger $\varepsilon$ reduces noise in continuous differential privacy mechanisms, PAS operates under discrete anchor-relative representations. Thus, ALE is governed not only by the anchor distance, but also by the discretization error introduced by the direction and distance bins. In this case, the nearest anchor does not necessarily reduce the reconstruction error, thus increasing $\varepsilon$ may concentrate probability mass on anchors whose induced uncertainty region is geometrically biased (suboptimal anchors). This can increase ALE and reduce Recall@k, resulting in non-monotonic privacy-utility behavior in our experiments.  

In general, the results show a clear trade-off. PAS sacrifices utility to achieve impressive coarse location privacy. However, more than half of the retrieval signal and ranking quality are maintained, while the generation performance is stable. These findings suggest that PAS enables practical privacy-preserving Spatial RAG, where strong privacy guarantees can be achieved without horrendous degradation in downstream tasks.

\section{Related Work}\label{sec:related_work}
\textbf{Retrieval-Augmented Generation (RAG)}:
Due to the parametric nature of LLMs, RAG has emerged as the dominant paradigm for grounding LLMs on external, up-to-date, and factual knowledge sources. Early studies \cite{lewis2020retrieval, kandpal2023large, gao2023retrieval} have established improved LLM performance in terms of factual accuracy and interpretability under the paradigm. Subsequent efforts focus on improving retrieval quality through methods such as dense representations \cite{karpukhin2020dense} and late interaction model \cite{khattab2020colbert}. Furthermore, \cite{guu2020retrieval, izacard2021leveraging} integrate retrieval into pre-training and generation pipelines. The RAG frameworks are increasingly being adopted in diverse application domains such as healthcare \cite{amugongo2025retrieval}, agriculture \cite{kumar2024overcoming}, and finance \cite{setty2024improving}. Importantly, due to the consistent increase in the ability of LLMs to reason over various types of data, including spatial data, RAGs are being integrated into georgraphical application domains, leading to the term \textit{Spatial RAG}. A growing body of work \cite{yu2025spatial, martins2025vision, jung2025intelligent, ni2025tp, zajac2025unifying, liu2025geospatial, schneider2025distrag, ruan2025retrieval} explores various  strategies to improve spatial RAG in various geographical application domains. Despite these advances, these works assume access to the users' true location, a vulnerability that presents significant risks and exposes users to location-based inference attacks.  
In the literature, several works have been proposed to protect user locations. Geo-indistinguishability \cite{andres2013geo} extends the differential privacy \cite{dwork2006calibrating, dwork2014algorithmic, mcsherry2007mechanism} to location information by adding noise to geographical coordinates. Other location protection approaches are spatial cloaking \cite{chow2011spatial, gruteser2003anonymous}  and k-anonymity \cite{gedik2005location, kim2021survey, kim2022privacy} that aim to protect user locations by expanding spatial regions. A more recent work \cite{lan2020trajectory} protects sequences of user movements instead of single points. However, these methods are not tailored for RAGs and are likely to degrade utility in the downstream tasks, due to precise spatial constraints associated with Spatial RAGs. 

\textbf{Privacy in Retrieval and LLM Systems}:
Recent research advances \cite{zeng2024good, edemacu2025privacy} have begun to address privacy concerns in LLM-based systems. Zhang et al. \cite{zhang2026privacy} proposes a framework that operates at the entity and relation granularity level using knowledge graph representations to protect the privacy of the retrieved documents during the generation phase. \cite{zeng2025mitigating} proposes using synthetic data as an alternative to sensitive KB data in RAGs. Duan et al. \cite{duan2024privacy} uses an ensembling approach by aggregating different versions of a prompted model to hide data. A significant portion of these studies have specifically focused on protecting sensitive user inputs. Carey et al. \cite{carey2024dp} privatizes tabular data with differential privacy before serializing them into a prompt. Liu et al. \cite{liu2026risk} employs differential privacy for token- and sentence-level sanitization on sensitive and labeled data in prompts. Ngong et al. \cite{ngong2025protecting} uses a locally deployable framework between users and LLM to identify and reformulate sensitive information in the user prompt before being sent to the LLM. However, all of these studies do not address spatial privacy.

Our proposed PAS framework bridges previously disconnected areas. Unlike previous approaches, the PAS framework introduces a structured probabilistic representation that enables retrieval under uncertainty. It integrates differential privacy with Spatial RAG, allowing retrieval without access to the true user location coordinates. By combining probabilistic anchor-based representations and hybrid retrieval, PAS enables privacy-preserving spatial reasoning while maintaining compatibility with Spatial RAG pipelines. To our knowledge, this is the first work to address privacy in Spatial RAGs.

\section{Conclusion}\label{sec:conclusion}
This work presented PAS, a structured mechanism to preserve location privacy in spatial retrieval within RAG systems. Unlike the conventional differential privacy based approaches that inject continuous noise, PAS encodes user location as a discrete anchor-relative representation comprising an anchor, and direction and distance bins. This design enables seamless integration into modern RAG pipelines while providing coarsely guaranteed location privacy. Through a comprehensive evaluation on a synthetic urban dataset, we showed that PAS achieves impressive location obfuscation while retaining over half of the baseline retrieval performance and a stable generation quality. A key insight from our analysis is that PAS exhibits a non-monotonic privacy-utility relationship with respect to the varying privacy parameter. Empirically, we showed that this behavior is due to geometrical biases induced by anchor discretization and binning, instead of purely stochastic noise. In general, our results demonstrate that PAS provides ingredients for a practical spatial privacy in RAG systems. 

\textbf{Limitations}: We acknowledge that this work has few limitations that we shall aim to address in the future. (i) PAS relies on anchor-based discretizations, which introduce geometric bias that can move the inferred region away from the user's true location in a structured manner. This can lead to the non-monotonic behavior discussed earlier, and it may affect robustness in sparse or irregular anchor regions. (ii) the use of Monte Carlo sampling for estimating the spatial relevance introduces computational overhead and approximation variance. Although suitable for lab experiments, it may not scale well for larger datasets and real-time environments. (iii) we evaluate PAS on synthetic urban data, even though it is controlled and diverse, it may not fully reflect the real-world distributions and noisy metadata. Thus, further  validations are required to generalize the results for production-scale services.

\medskip

{
\small
\bibliographystyle{IEEEtran}
\bibliography{ref.bib}



}


\appendix
\section*{Appendices}

\section{PAS vs Baseline Graphical Comparison}\label{app:pas_baseline}
We also provide a graphical retrieval and generation performance comparison between the baseline and PAS mechanism as presented in \ref{fig:graph_compare}. We observe that PAS induces drop in retrieval and generation performance. However, the drop is more for retrieval compared to downstream generation. This visually confirms the utility results presented in Table \ref{tab:pas_results}. 
\begin{figure*}[htbp]
     \centering
     \begin{subfigure}[b]{0.48\textwidth}
         \centering
         \includegraphics[width=\textwidth]{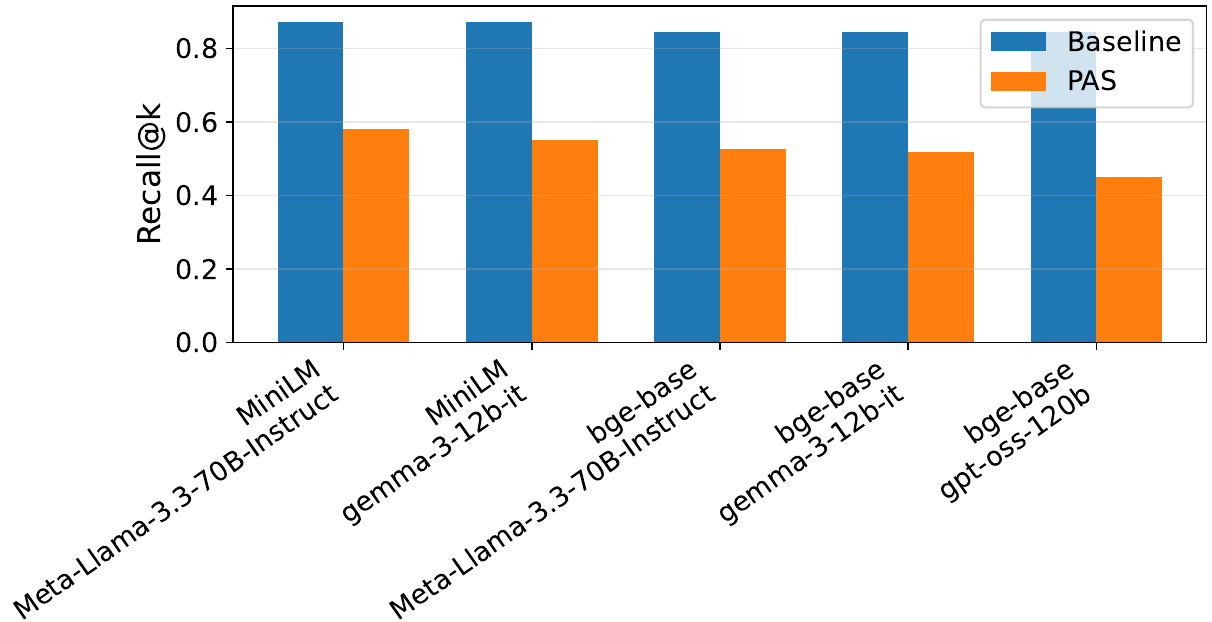}
         \caption{Recall}
         \label{fig:sub1}
     \end{subfigure}
     \hfill
     \begin{subfigure}[b]{0.48\textwidth}
         \centering
         \includegraphics[width=\textwidth]{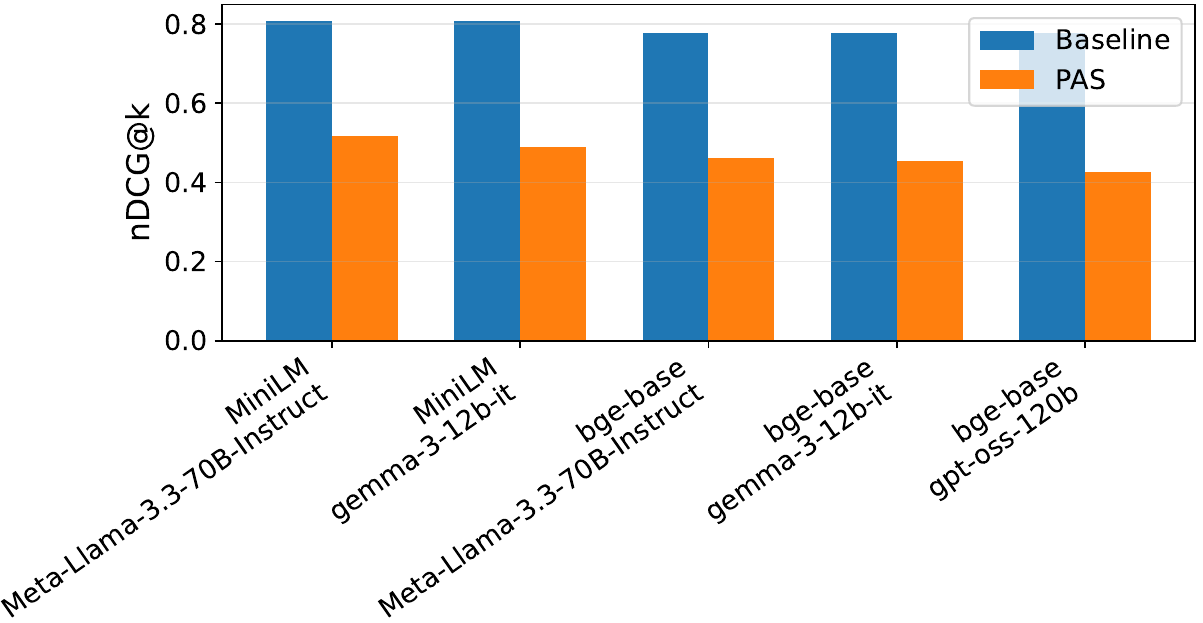}
         \caption{nDCG}
         \label{fig:sub2}
     \end{subfigure}
     \hfill
     \begin{subfigure}[b]{0.48\textwidth}
         \centering
         \includegraphics[width=\textwidth]{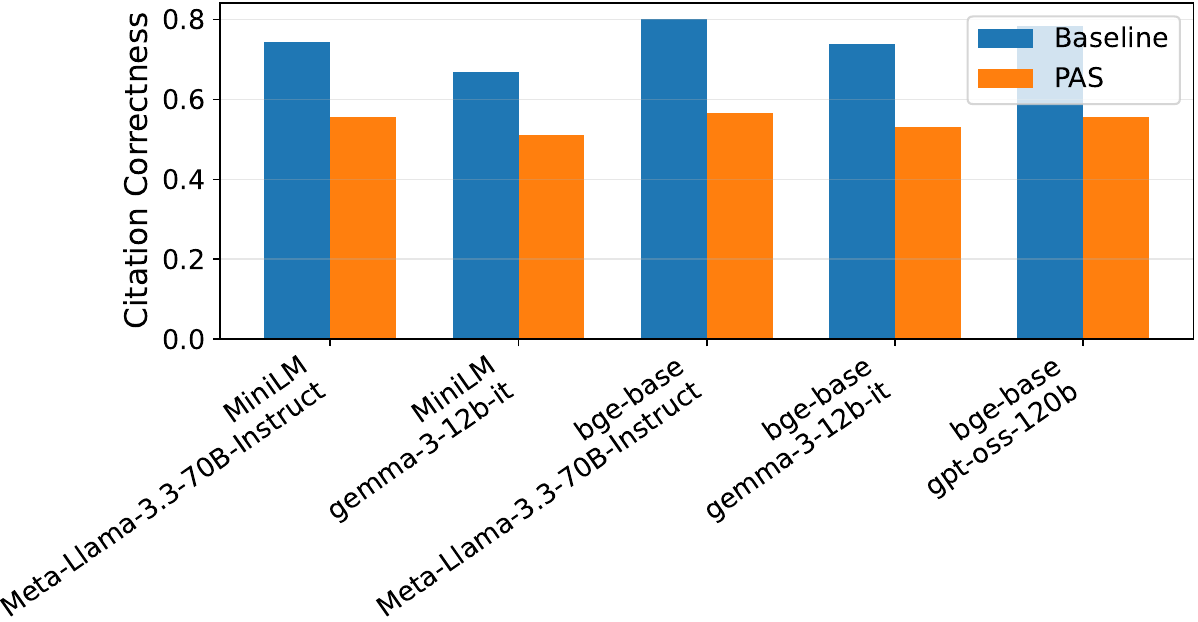}
         \caption{Citation}
         \label{fig:sub2}
     \end{subfigure}
     \hfill
     \begin{subfigure}[b]{0.48\textwidth}
         \centering
         \includegraphics[width=\textwidth]{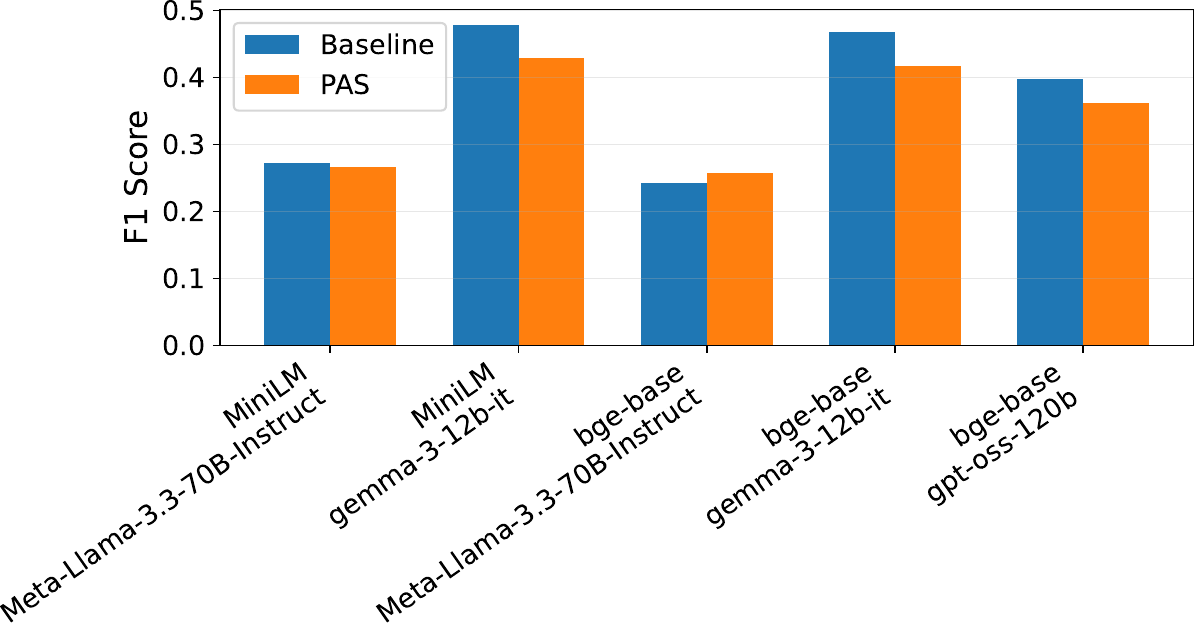}
         \caption{F1 Score}
         \label{fig:sub2}
     \end{subfigure}
     \caption{PAS Performance Comparison with the Baseline}
     \label{fig:graph_compare}
\end{figure*}

\section{Retrieval vs Generation}\label{app:retrieval_generation} 
Similarly, we evaluate how retrieval affects generation. We perform this by computing generation metrics (citation correctness and F1 score) under varying retrieval metrics (Recall@k and nDCG@k) for different retriever-generation model configurations as presented in Figure \ref{fig:retrieve_gen}. In general, we observe that the bge-base retriever performs worse in comparison to MiniLM, and the retrieval performance has slight impact on citation correctness but close to negligible performance on F1 scores. This confirms our earlier claim that the model can still produce quality results with imperfect retrievals.
\begin{figure*}[htbp]
     \centering
     \begin{subfigure}[b]{0.48\textwidth}
         \centering
         \includegraphics[width=\textwidth]{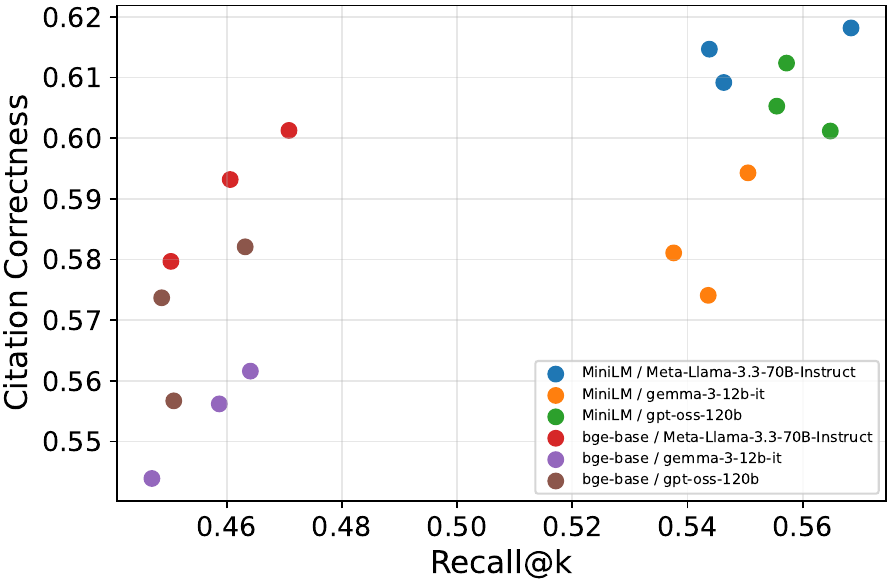}
         \caption{Recall vs Citation}
         \label{fig:sub1}
     \end{subfigure}
     \hfill
     \begin{subfigure}[b]{0.48\textwidth}
         \centering
         \includegraphics[width=\textwidth]{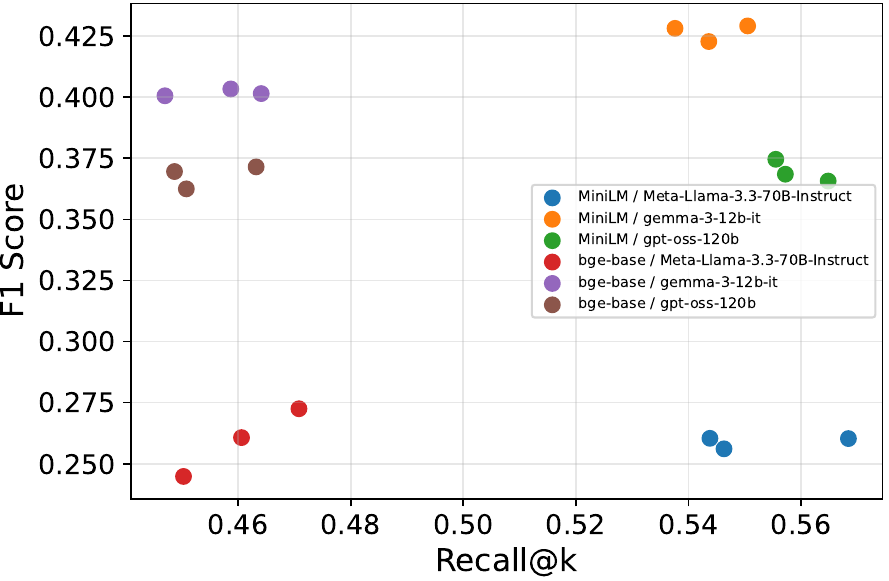}
         \caption{Recall vs F1}
         \label{fig:sub2}
     \end{subfigure}
     \hfill
     \begin{subfigure}[b]{0.48\textwidth}
         \centering
         \includegraphics[width=\textwidth]{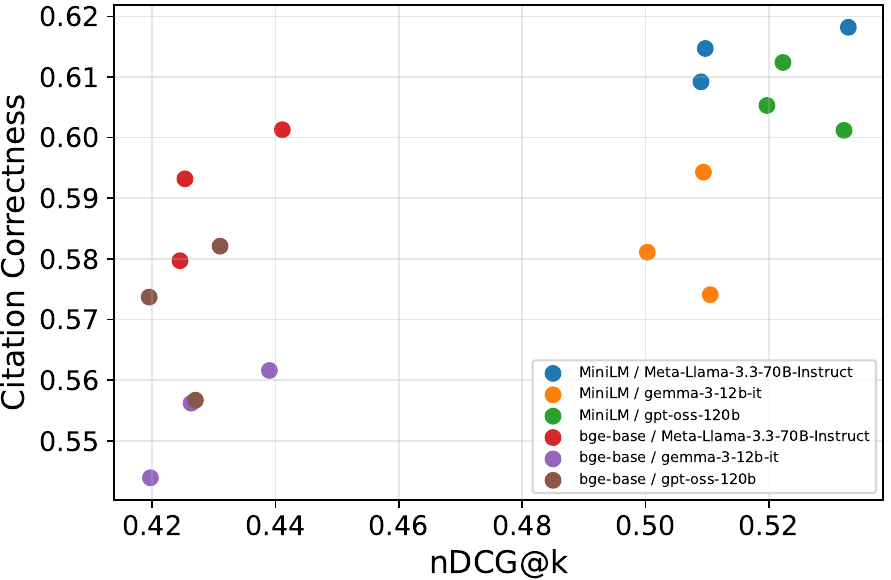}
         \caption{nDCG vs Citation}
         \label{fig:sub2}
     \end{subfigure}
     \hfill
     \begin{subfigure}[b]{0.48\textwidth}
         \centering
         \includegraphics[width=\textwidth]{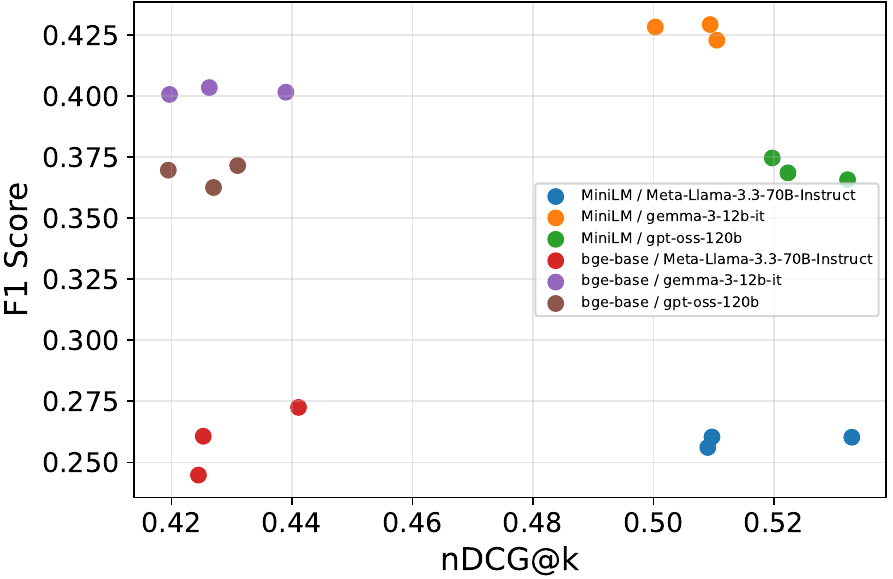}
         \caption{nDCG vs F1}
         \label{fig:sub2}
     \end{subfigure}
     \caption{Retrieval vs Generation Performance}
     \label{fig:retrieve_gen}
\end{figure*}

\section{Experimental Setup}\label{app:experiment_setup}
In this section, we present additional details on our experimental setup. For reproducibility, we will release our implementation codes and dataset through the GitHub repository once the work is accepted for publication.

\subsection{Retriever Configuration}

We employ dense retrieval models to capture semantic relevance between queries and candidate documents. Specifically, we use two widely adopted embedding models: \texttt{sentence-transformers/all-MiniLM-L6-v2}~\cite{reimers2019sentencebert} 
and \texttt{BAAI/bge-base-en-v1.5}~\cite{xiao2023bge}. For each query--document pair, semantic similarity is computed using cosine similarity in the embedding space.

To incorporate spatial relevance, we use the hybrid scoring function defined in Equation~\ref{eqn:hybrid_score}, which combines semantic and spatial signals. In the main experiments, we retrieve the top-\(k=5\) candidates based on this combined score. 

\subsection{Privacy-Aware Spatial (PAS) Mechanism}

We implement the PAS mechanism described in Section~\ref{sec:pas_mapping} by instantiating the probabilistic anchor selection in Eq.~\ref{eq:selection} and the resulting uncertainty region \(U(z)\). In practice, this mechanism is realized through discrete anchor sets and binning functions for direction and distance, as defined earlier.

For all experiments, we use a fixed set of public anchors. The exponential mechanism in Equation~\ref{eq:selection} is applied with \(\varepsilon \in \{1.0, 2.0, 5.0\}\) and a scaling factor \(\texttt{s} = 500\).

Directional and distance discretization follows the binning scheme defined in Section~\ref{sec:pas_mapping}, with 8-way directional bins (\texttt{N, NE, E, SE, S, SW, W, NW}) and distance bins (\texttt{0-0.5mi, 0.5-1mi, 1-2mi, 2mi+}). Given a sampled anchor \(a^*\), we compute the corresponding token \(z = (a^*, B_{\theta}^*, B_r^*)\) using the deterministic post-processing procedure described previously.

At retrieval time, instead of operating on a single estimated user location, we approximate the distribution \(\pi(x \mid z)\) over the uncertainty region \(U(z)\) using Monte Carlo sampling, as defined in Eq.~\ref{eq:monte}. Specifically, we draw \(K=1000\) latent user samples within \(U(z)\) and use these samples to estimate the spatial score \(S_{sp}(t \mid z)\) via the Monte Carlo formulation in Eq.~\ref{eq:monte}.

In implementation, the PAS token also carries auxiliary metadata, including anchor identifiers and sampling probabilities, to ensure consistent filtering during retrieval (e.g., restricting candidate anchors when required). These components collectively enable the retriever to operate over a privacy-preserving representation of user location while remaining compatible with the hybrid scoring framework in Eq.~\ref{eqn:hybrid_score}.

\subsection{Baseline Configuration}

As a point of comparison, we consider a non-private baseline that assumes access to the user’s exact location. The baseline uses the same dense retrievers and identical score fusion formulation, differing only in the absence of PAS perturbation. This ensures that performance differences can be attributed directly to the privacy mechanism.

We evaluate retrieval and generation performance using multiple metrics, including Recall@k, nDCG@k, citation correctness, and F1 score.

\subsection{LLM-Based Generation}

For answer generation, we use multiple large language models (LLMs), namely \texttt{gpt-oss-120b}, \texttt{Meta-Llama-3.3-70B-Instruct}~\cite{touvron2023llama}, 
and \texttt{gemma-3-12b-it}~\cite{gemma2024}. Generation is guided by a structured prompt template shown in Figure~\ref{fig:prompt}, in which a system prompt enforces strict grounding in the retrieved context, and a user prompt provides the query, semantic targets (e.g., entities or tags), spatial intent (direction and radius), and retrieved documents. All models are decoded using low-temperature sampling (\(T = 0.1\)) to reduce variability and improve determinism.

\begin{figure}[t]
\centering
\footnotesize

\begin{minipage}{\columnwidth}
\fbox{
\begin{minipage}{0.98\columnwidth}
\ttfamily

\textbf{System Prompt}\\
You are a careful RAG answering assistant.\\
Answer the user's query using ONLY the retrieved context.\\
Do not use outside knowledge.\\
If the evidence is weak or partially mismatched, say so briefly.\\
Prefer concise answers that directly name the best matching places.\\
Only cite documents that actually support your answer.\\
Return STRICT JSON with exactly these keys:\\
\{\\
\ \ "answer": string,\\
\ \ "citations": [{"title": string, "doc\_id": string}],\\
\ \ "faithfulness\_notes": [string]\\
\}\\
Do not include markdown fences.\\

\vspace{0.5em}

\textbf{User Prompt Template}\\
User query: \{raw\_query\}\\

Semantic target:\\
- entity\_type: \{entity\}\\
- must\_have\_tags: \{must\_have\_tags\}\\

Spatial intent:\\
- direction\_constraint: \{direction\_constraint\}\\
- radius\_miles: \{radius\_miles\}\\

Retrieved context:\\
\{context\_text\}\\

Instructions:\\
1. Write a grounded answer that directly answers the user query.\\
2. Mention strongest matching places first.\\
3. Exclude unsupported claims.\\
4. Mention uncertainty for weak spatial matches.\\
5. Cite only docs used in the answer.\\
6. Output STRICT JSON only.

\end{minipage}
}
\end{minipage}

\caption{Our prompt template used for grounded generation.}
\label{fig:prompt}
\end{figure}

\subsection{Compute and Hardware Details}\label{app:hardware}
We provide details of the computational  resources used to perform our experiments.
All embedding-based retrieval tasks using the retrieval models were conducted on a single NVIDIA A40 GPU. The GPU was sufficient to perform all the retrieval tasks for the dataset scale used in the work. The generator models were not hosted locally, but were accessed via an API-based interface provided by SambaNova\footnote{\url{https://sambanova.ai/}}. All the generation results were obtained through these API calls. This introduced network latency, but it does not affect the correctness of the results.

\newpage
\input{checklist.tex}

\end{document}

%% file: checklist.tex
\section*{NeurIPS Paper Checklist}

\begin{enumerate}

\item {\bf Claims}
    \item[] Question: Do the main claims made in the abstract and introduction accurately reflect the paper's contributions and scope?
    \item[] Answer: \answerYes{}
    \item[] Justification: 
    The abstract clearly states the core contributions of the PAS framework, including, the structure achor-based spatial privacy mechanism, experimental evaluation indicating $\sim$370-400m adversarial localization error with retention of over half of the baseline utility performance and experimental analysis of the non-monotonic privacy-utility relationship. These claims are consistent with the experimental results presented later in the paper.

\item {\bf Limitations}
    \item[] Question: Does the paper discuss the limitations of the work performed by the authors?
    \item[] Answer: \answerYes{}
    \item[] Justification: 
    The paper explicitly indicates the limitations within Section \ref{sec:conclusion} (conclusion section), acknowledging the limitations of the proposed PAS framework. These limitations include: geometric bias due to anchor discretizations, computational overhead and variance due to Monte Carlo sampling, and use of synthetic data. The discussions offer potential implications and future directions.

\item {\bf Theory assumptions and proofs}
    \item[] Question: For each theoretical result, does the paper provide the full set of assumptions and a complete (and correct) proof?
    \item[] Answer: \answerNo{}
    \item[] Justification: 
    The paper presents a theoretical result (Lemma 1), stating that PAS satisfies $\varepsilon$-geo-indistinguishability. While the paper presents the mechanism, assumptions and an intuitive argument, it only contains a proof sketch. This is because, the proof relies on the standard properties of exponential mechanism and post-processing invariance, and thus, no step-by-step derivation for complete proof.

    \item {\bf Experimental result reproducibility}
    \item[] Question: Does the paper fully disclose all the information needed to reproduce the main experimental results of the paper to the extent that it affects the main claims and/or conclusions of the paper (regardless of whether the code and data are provided or not)?
    \item[] Answer: \answerYes{}
    \item[] Justification: 
    The paper provides detailed information to reproduce the main experimental results. The dataset construction is fully described. The PAS mechanism is fully specified from the probabilistic anchor selection to latent user sampling. The retrieval pipeline is clearly defined including the retrievers and generators used. The evaluation metrics are explicitly defined, and the used hyperparameters are reported.

\item {\bf Open access to data and code}
    \item[] Question: Does the paper provide open access to the data and code, with sufficient instructions to faithfully reproduce the main experimental results, as described in supplemental material?
    \item[] Answer: \answerYes{} 
    \item[] Justification: 
    All the resource will be shared publicly via a repository once the paper is accepted for publication. The repository will contain the dataset and implementation codes used during the experiment. Particularly, it will include: the dataset files used for benchmarking, the end-to-end PAS pipeline implementations, the configurations that enable reproduction of the main results and an example illustration of the PAS mechanism at work. Together, these resources will allow a reader to reproduce the reported experimental results.

\item {\bf Experimental setting/details}
    \item[] Question: Does the paper specify all the training and test details (e.g., data splits, hyperparameters, how they were chosen, type of optimizer) necessary to understand the results?
    \item[] Answer: \answerNo{}
    \item[] Justification: 
    The paper provides significant amount of details about the experimental setup, including dataset construction, retrieval pipeline, PAS mechanism and key parameters. However, the work does not include training a new model, but it uses existing retrievers and generation models. The configuration details will be included in the source code once the work is accepted for publication.

\item {\bf Experiment statistical significance}
    \item[] Question: Does the paper report error bars suitably and correctly defined or other appropriate information about the statistical significance of the experiments?
    \item[] Answer: \answerNo{}
    \item[] Justification: 
    The paper reports mean performance metrics across evaluation queries. It does not include error bars, confidence intervals, or statistical tests such as standard deviation, bootstrap intervals, and hypothesis testing. The experiments are conducted over large evaluation sets (423 queries) and report aggregate metrics that exhibit consistent and substantial differences between the baseline and the PAS. Such effect sizes are significantly larger than typical variances introduced by sampling noise. 

\item {\bf Experiments compute resources}
    \item[] Question: For each experiment, does the paper provide sufficient information on the computer resources (type of compute workers, memory, time of execution) needed to reproduce the experiments?
    \item[] Answer: \answerYes{} 
    \item[] Justification: 
    Much as the paper focuses on the methodology, dataset, and evaluation pipelines. It specifies the retrievers and generation models used, it provides detailed information on the computational resources used in Section \ref{app:hardware}.
    
\item {\bf Code of ethics}
    \item[] Question: Does the research conducted in the paper conform, in every respect, with the NeurIPS Code of Ethics \url{https://neurips.cc/public/EthicsGuidelines}?
    \item[] Answer: \answerYes{}
    \item[] Justification: 
    The research complies with NeurIPS Code of Ethics in all important aspect. The dataset used is synthetic, thus avoids the use of real user location data and thereby eliminating the concerns associated with personal data misuse and privacy violations. The methodology does not involve use of human subjects, sensitive attribute and potentially harmful application. Finally, no deceptive and unsafe practices are employed.

\item {\bf Broader impacts}
    \item[] Question: Does the paper discuss both potential positive societal impacts and negative societal impacts of the work performed?
    \item[] Answer: \answerYes{}
    \item[] Justification: 
    The paper explicitly emphasizes the societal impact of PAS framework, particularly, it enables privacy-preserving retrieval, thus reducing the risks associated with location tracking and data exposure. 
    
\item {\bf Safeguards}
    \item[] Question: Does the paper describe safeguards that have been put in place for responsible release of data or models that have a high risk for misuse (e.g., pre-trained language models, image generators, or scraped datasets)?
    \item[] Answer: \answerNA{}
    \item[] Justification: 
    The paper does not introduce or release any high-risk models or datasets.

\item {\bf Licenses for existing assets}
    \item[] Question: Are the creators or original owners of assets (e.g., code, data, models), used in the paper, properly credited and are the license and terms of use explicitly mentioned and properly respected?
    \item[] Answer: \answerYes{}
    \item[] Justification: 
    The paper appropriately credits all assets used in the work. The pretrained models e.g., bge-base, MiniLM, Meta-Llama-3.3-70B, gemma-3-12b-it, and gpt-oss-120b are all identified and cited. The dataset is synthetically generated, avoiding reliance on third party proprietory data and associated license issues. 
    The code has proper attributions to dependencies and libraries.
    No restricted asset is used without acknowledgment. 

\item {\bf New assets}
    \item[] Question: Are new assets introduced in the paper well documented and is the documentation provided alongside the assets?
    \item[] Answer: \answerYes{}
    \item[] Justification: 
    The paper introduces new assets such as the synthetic spatial dataset and the PAS implementation. These assets will be released via public repositories once the work is accepted for publication. 
    They will be accompanied by descriptive metadata and schema. An illustrative example will also be added to demonstrate PAS at work. 

\item {\bf Crowdsourcing and research with human subjects}
    \item[] Question: For crowdsourcing experiments and research with human subjects, does the paper include the full text of instructions given to participants and screenshots, if applicable, as well as details about compensation (if any)? 
    \item[] Answer: \answerNA{}
    \item[] Justification: 
    The paper does not involve crowdsourcing experiments or human subjects. All the experiments are conducted over a synthetic dataset.

\item {\bf Institutional review board (IRB) approvals or equivalent for research with human subjects}
    \item[] Question: Does the paper describe potential risks incurred by study participants, whether such risks were disclosed to the subjects, and whether Institutional Review Board (IRB) approvals (or an equivalent approval/review based on the requirements of your country or institution) were obtained?
    \item[] Answer: \answerNA{} 
    \item[] Justification: 
    The paper does not involve human subjects and all the experiments are conducted over a synthetic dataset. Thus, there are no participant risks to disclose and IRB approval is not required.

\item {\bf Declaration of LLM usage}
    \item[] Question: Does the paper describe the usage of LLMs if it is an important, original, or non-standard component of the core methods in this research? Note that if the LLM is used only for writing, editing, or formatting purposes and does \emph{not} impact the core methodology, scientific rigor, or originality of the research, declaration is not required.
    \item[] Answer: \answerYes{}
    \item[] Justification: 
    The paper clearly describes the role of LLMs as part of the evaluation pipeline, where they are used for answer generation in the Spatial RAG setting. Additionally, the paper acknowledges that ChatGPT was used during dataset construction, specifically to assist in generating components of the synthetic queries and metadata. This usage does not affect the PAS methodology but it is important for transparency.

\end{enumerate}